\documentclass[11pt]{llncs}
\usepackage[a4paper,top=2cm,bottom=2cm,left=2.5cm,right=2.5cm]{geometry}
\usepackage[T1]{fontenc}
\usepackage{graphicx}
	\graphicspath{{./figs}}
	\DeclareGraphicsExtensions{.pdf,.png}
\usepackage{hyperref}
	\usepackage{color}
	
\usepackage[english]{babel}
\usepackage{amssymb,amsfonts}
\usepackage{algpseudocodex}
\usepackage{algorithm}
\usepackage{booktabs}
	\newcommand{\ra}[1]{\renewcommand{\arraystretch}{#1}}
\usepackage{mathtools}
\usepackage{xspace}
\usepackage{textgreek}
\usepackage{newtxtext,newtxmath}
\usepackage{microtype}
\usepackage{tikz}
	\usetikzlibrary{arrows,automata,shapes}
	\usetikzlibrary{positioning}
	\tikzstyle{accepting}=[double distance=1.5pt, outer sep=1pt+\pgflinewidth]

\newcommand{\Ptime}{\text{P}\xspace}
\newcommand{\NP}{\text{NP}\xspace}
\newcommand{\coNP}{\text{coNP}\xspace}
\newcommand{\PH}{\text{PH}\xspace}
\newcommand{\SigmaP}{\Sigma^P\xspace}
\newcommand{\PiP}{\Pi^P\xspace}

\DeclareMathOperator{\A}{\mathcal A}
\DeclareMathSymbol{\shortminus}{\mathbin}{AMSa}{"39}

\newtheorem{hypothesis}[theorem]{Hypothesis}
\begin{document}
\title{On the Secret Protection Problem in~Discrete-Event Systems}
\author{Tom{\' a}{\v s}~Masopust \and Jakub Ve{\v c}e{\v r}a}
\authorrunning{T.~Masopust, J.~Ve{\v c}e{\v r}a}
\institute{Faculty of Science, Palacky University Olomouc, Czechia\\
\email{tomas.masopust{@}upol.cz}, \email{jakub.vecera01@upol.cz}\\
\url{https://apollo.inf.upol.cz:81/vecerajakub/SPP-automata}
}
\maketitle
\begin{abstract}
	The secret protection problem (SPP) seeks to synthesize a minimum-cost policy ensuring that every execution from an initial state to a secret state includes a sufficient number of protected events. The problem is solvable in polynomial time under the assumption that transitions are uniquely labeled. When this assumption is relaxed, the problem becomes weakly \NP-hard. We first strengthen the result by showing that the problem is strongly \NP-hard even if all parameters are restricted to binary values. We then propose a formulation of SPP as an integer linear programming (ILP) problem, and empirically evaluate the scalability and effectiveness of the ILP-based approach on relatively large systems. Finally, we examine the complexity of a variant of SPP in which only distinct protected events contribute to clearance and show that its decision version is $\Sigma_{2}^{P}$-complete.
\end{abstract}

\section{Introduction}
	Modern systems frequently store and process sensitive information, from personal data to financial credentials, which must remain confidential to unauthorized users. This concern led to the formulation of the \emph{Secret Protection Problem} (SPP), a framework for guaranteeing that users complete a predefined number of security checks before gaining access to sensitive parts of a system, see \cite{MatsuiC19,MaC22,MaC25} for further details, including discussions on the motivation and applications.
	
	In the SPP framework, a system is modeled as a finite automaton with a designated set of \emph{secret states}. Each secret state is associated with a \emph{security level}, defined as a non-negative integer quantifying the minimum number of security checks required to reach that state. The event set is partitioned into \emph{protectable} and \emph{unprotectable} events. Each protectable event is assigned a non-negative \emph{cost}, representing the implementation effort of the corresponding security check, and a \emph{clearance level}, indicating the number of clearance units granted upon its execution. An event is \emph{protected} if a security check has been placed on it.
	The objective of SPP is to synthesize a \emph{protecting policy\/} of minimal cost such that every execution path from an initial state to a secret state includes a number of protected events that meets or exceeds the security level assigned to the respective secret state.
	This problem is referred to as Parikh-SPP in this paper.
	
	If all transitions are labeled with unique events, Ma and Cai~\cite{MaC22} showed how to solve Parikh-SPP via a reduction to a supervisory control problem or to a max-flow network problem, providing thus a polynomial-time algorithm. If the assumption is relaxed, Ma et al.~\cite{MaJC24} proved that Parikh-SPP becomes weakly \NP-hard. 
	However, they left two questions open: \emph{Is the problem strongly \NP-hard?} and \emph{What is the complexity of a special case of Parikh-SPP, called the uniform Parikh-SPP (Parikh-SPP-U) in which all events have a uniform clearance level of one?}

	We answer both problems by proving that Parikh-SPP\nobreakdash-U is strongly NP\nobreakdash-hard, and that it remains \NP-hard even if the ranges of the cost function, the clearance function, and the security requirement function are restricted to be binary. Furthermore, we show that no sub-exponential-time algorithm exists for Parikh-SPP unless the Exponential Time Hypothesis~\cite{ImpagliazzoP01} fails.
	
	These intractability results raise a fundamental question: \emph{How can Parikh-SPP be practically solved for real-world systems?} Although Ma et al.~\cite{MaJC24} proposed a polynomial-time heuristic algorithm, its practical applicability remains limited---instances with only 100 states can require up to five minutes to solve, and the algorithm offers no guarantees on the quality of the solution. 
	This gap highlights the need for a robust and efficient method to solve Parikh-SPP. We address this challenge by formulating Parikh-SPP as an integer linear programming (ILP) problem and demonstrating its scalability and effectiveness in finding exact solutions for relatively large systems across an extensive set of benchmarks, including both real-world and randomly generated data.
	
	Finally, we analyze the complexity of a variant of SPP in which only distinct protected events contribute to the clearance, referred to as Indicator-SPP. This variant is motivated by real-world multi-factor authentication schemes. We prove that its decision version is $\Sigma_{2}^{P}$-complete, placing it at the second level of the polynomial-time hierarchy. In contrast to Parikh-SPP, the existence of a practically efficient algorithm for Indicator-SPP remains an open problem.

\section{Preliminaries}
	We assume that the reader is familiar with automata theory~\cite{sipser1996introduction}.
	The sets of non-negative integers and non-negative real numbers are denoted by $\mathbb{N}$ and $\mathbb{R}_{\ge 0}$, respectively. For a set $S$, the cardinality of $S$ is denoted by $|S|$, and its power set by $2^S$. If $S=\{x\}$ is a singleton, we often write $x$ instead of $\{x\}$.
	An alphabet $\Sigma$ is a finite nonempty set of \emph{events}. A string over $\Sigma$ is a finite sequence of events from $\Sigma$. The set of all strings over $\Sigma$ is denoted by $\Sigma^*$, and the empty string is denoted by $\varepsilon$. A language $L$ over $\Sigma$ is a subset of $\Sigma^*$. For a string $u\in \Sigma^*$ and an event $a\in \Sigma$, the number of occurrences of $a$ in $u$ is denoted by~$|u|_a$.

	A {\em nondeterministic finite automaton\/} (NFA) is a quintuple $\A = (Q,\Sigma,\delta,I,F)$, where $Q$ is a finite set of states, $\Sigma$ is an input alphabet, $I\subseteq Q$ is a set of initial states, $F \subseteq Q$ is a set of accepting states, and $\delta \colon Q\times\Sigma \to 2^Q$ is a transition function that can be inductively extended to $\delta\colon 2^Q\times\Sigma^* \to 2^Q$. 
	The language generated by $\A$ is the set $L(\A) = \{w\in \Sigma^* \mid \delta(I,w) \neq\emptyset\}$, and the language accepted by $\A$ is the set $L_m(\A) = \{w\in \Sigma^* \mid \delta(I,w)\cap F \neq\emptyset\}$.

	We briefly review the concepts from complexity theory~\cite{AroraBarak2009}. A \emph{decision problem\/} is a yes--no question defined over instances of a problem. The problem is \emph{decidable\/} if there is an algorithm that determines, for every instance, whether the answer is yes or no.
	Decidable problems are categorized into complexity classes based on the time or space required to solve them. The classes \Ptime and \NP consist of problems solvable by deterministic and nondeterministic polynomial-time algorithms, respectively. A problem is \NP-complete if it satisfies two conditions:
	(i) Membership---the problem belongs to \NP, and
	(ii) Hardness---every problem in \NP can be reduced to it in deterministic polynomial time.
	While $\Ptime \subseteq \NP$, whether the inclusion is strict is one of the most important open problems in computer science.
	
	An \NP-complete problem is \emph{weakly} \NP-complete if it is solvable in pseudo-polynomial time, i.e., in time polynomial in the numeric value of the input rather than in the length of the input (the number of bits representing it). A problem is \emph{strongly} \NP-complete if it remains \NP-hard even if the data are encoded in unary, i.e., independent of the numerical values of the data~\cite{GareyJ79}. Hereafter, \NP-complete means strongly \NP-complete unless stated otherwise.

	The levels of the \emph{poly\-no\-mi\-al-time hierarchy} (\PH) are defined as follows: $\SigmaP_0 = \PiP_0 = \Ptime$, and for $k \geq 0$, the class $\SigmaP_{k+1}$ contains problems solvable by a nondeterministic polynomial-time Turing machine with access to an oracle for a problem in $\SigmaP_k$. Analogously, $\PiP_{k+1}$ is  the class of problems whose complements are in $\SigmaP_{k+1}$. In particular, $\SigmaP_1=\NP$ and $\PiP_1 = \coNP$, that is, $\coNP$ is the class of problems whose complements are in $\NP$. It is widely believed that the hierarchy is strict, see~\cite{Stockmeyer76} for more details. 

	A {\em Boolean formula\/} consists of variables, logical connectives ($\land, \lor, \neg$), and parentheses. A {\em literal\/} is a variable or its negation. A {\em clause\/} is a disjunction of literals. A formula is in {\em conjunctive normal form\/} (CNF) if it is a conjunction of clauses. If each clause contains at most $k$ literals, the formula is in $k$-CNF.
 	Analogously, a \emph{conjunct} is a conjunction of literals, and a formula is in \emph{disjunctive normal form} (DNF) if it is a disjunction of conjuncts. If each conjunct contains at most $k$ literals, the formula is in $k$-DNF.
	A formula is {\em satisfiable\/} if there is an assignment of 1 and 0 to the variables such that the formula evaluates to 1.
	For $k\ge 1$, the $k$-CNF {\em Boolean satisfiability problem\/} ($k$-SAT) asks whether a given formula in $k$-CNF is satisfiable.
	If a formula in $k$-CNF has $n$ variables, enumerating possible truth assignments results in an $O(2^n n^k)$-time algorithm for $k$-SAT.
	\emph{Quantified Boolean formulas} (QBFs) extend Boolean formulas by allowing quantification over variables.
	
	3-SAT is \NP-complete, and despite extensive efforts, no sub-exponential-time algorithm has been found to date. The lack of progress led to the formulation of the exponential time hypothesis (ETH), which posits that 3-SAT cannot be solved in sub-exponential time $2^{o(n)}$, where $n$ is the number of variables in the 3-CNF formula; recall that $o(g(n))$ denotes the class of functions $f(n)$ that grow asymptotically much slower than $g(n)$, i.e., $\lim_{n\to \infty} f(n)/g(n) = 0$.
	\begin{hypothesis}[ETH, \cite{ImpagliazzoP01}]
		There is a $\lambda > 0$ such that 3\nobreakdash-SAT cannot be solved in time $O(2^{\lambda n})$, where $n$ is the number of variables in the formula.
	\end{hypothesis}
	Note that ETH allows for algorithms that solve 3-SAT in time $O(c^n)$ for $c<2$. In fact, the current fastest algorithm by Paturi et al.~\cite{PaturiPSZ05}, improved by Hertli~\cite{Hertli14}, runs in time $O^*(1.30704^n)$; we use the soft big-$O$ notation $O^*$ to disregard polynomial factors $n^{O(1)}$ in the big-$O$ notation; e.g., we can write the complexity $O(2^n n^k)$ as $O^*(2^n)$.

\section{The Secret Protection Problem}
	In this paper, the formulation of the problem is intentionally simplified for clarity of our generalization. For the original treatment, we refer to Ma and Cai~\cite{MaC22} or Ma et al.~\cite{MaJC24}.

	Let $\A = (Q,\Sigma,\delta,I,S)$ be an NFA. States of $S$ are called \emph{secret states}. A {\em security requirement\/} is a function \mbox{$\ell\colon Q \to \mathbb{N}$} that assigns a non-negative security level $\ell(q)$ to each secret state $q\in S$. Intuitively, the security level specifies the minimum number of \emph{protected events} that must be traversed to reach state $q$ from any initial state. We assume that initial states are not secret, i.e., $I\cap S = \emptyset$.
	
	The event set $\Sigma$ is partitioned into {\em protectable events\/} $\Sigma_p$ and {\em unprotectable events\/} $\Sigma_{up}$; we use the notation $\Sigma = \Sigma_p \uplus \Sigma_{up}$. The protectable events are characterized by two functions:
	\begin{itemize}
		\item The \emph{clearance function} $\gamma \colon \Sigma_p \to \mathbb{N}$ specifying the number of clearance units granted upon the execution of a protectable event.
		
		\item The \emph{cost function} $c \colon \Sigma_p \to \mathbb{R}_{\geq 0}$ assigning a cost to the implementation of protection for each protectable event. 
	\end{itemize}
 
	For every protectable event $a\in \Sigma_p$, let $\chi_a \colon \Sigma^* \to \mathbb{N}$ be a function that assigns a natural number to each string over $\Sigma$. We denote by $\chi = \{\chi_a \mid a \in \Sigma_p\}$ the family of these functions. In this paper, we focus on two specific instances:
	\begin{enumerate}
		\item the \emph{Parikh function} $\chi_a(w) = |w|_a$, which counts every occurrence of $a$ in $w$, and
		\item the \emph{indicator function} 
		\[
			\chi_a(w) = 
				\begin{cases}
					1 & \text{ if $a$ occurs in $w$,}\\ 
					0 & \text{ otherwise,}
				\end{cases}
		\]
		which counts each protectable event in $w$ only once.
	\end{enumerate}
 
 	A \emph{protecting policy} $\mathcal P \subseteq \Sigma_p$ is a subset of protectable events. A policy $\mathcal{P}$ is \emph{$\chi$-valid} if, for every string $w \in L(\A)$ that leads to a secret state $q\in S$,
 	\[
 		 \sum_{a\in \mathcal P} \Bigl(\gamma(a) \cdot \chi_a(w)\Bigr) \ge \ell(q)\,.
 	\]
	The \emph{cost of a policy} $\mathcal{P}$, denoted by $C(\mathcal P)$, is the sum of the costs of all events included in the policy:
	\[
		C(\mathcal{P}) = \sum_{a \in \mathcal{P}} c(a)\,.
	\] 
	A policy is $\chi$-\emph{optimal} if it is $\chi$-valid and no other $\chi$-valid policy has a lower cost.
	The \emph{secret protection problem under $\chi$-optimality} is defined as follows.	
	\begin{definition}[$\chi$-SPP]
		Given a system $\A = (Q, \Sigma, \delta, I, S)$ with $\Sigma = \Sigma_p \uplus \Sigma_{up}$, along with a security requirement $\ell$, a clearance function $\gamma$, and a cost function $c$, the objective is to determine a $\chi$-optimal protection policy $\mathcal P\subseteq \Sigma_p$.
	\end{definition}

	A special case of $\chi$-SPP, the \emph{uniform $\chi$-SPP} ($\chi$-SPP-U), simplifies the clearance mechanism by assuming that the clearance level for every protectable event is uniformly set to one, i.e., $\gamma(a) = 1$ for all $a \in \Sigma_p$. 
	
	We now formulate the decision versions, which generalize the decision version presented by Ma et al.~\cite{MaJC24}. 
	
	\begin{definition}[BC-$\chi$-SPP]
		The \emph{budget-constrained $\chi$-SPP} (BC\nobreakdash-$\chi$-SPP) asks whether, given a system $\A$ over $\Sigma = \Sigma_p \uplus \Sigma_{up}$, a security requirement $\ell$, a clearance function $\gamma$, a cost function $c$, and a budget limit $W \in \mathbb{N}$, there is a valid protecting policy $\mathcal{P} \subseteq \Sigma_p$ such that the cost of the policy satisfies $C(\mathcal{P}) \leq W$. Analogously, we define BC-$\chi$-SPP-U.
	\end{definition}

	If $\chi$ consist purely of Parikh (resp. indicator) functions, we call the problem Parikh-SPP (resp. Indicator-SPP).
	
	The concept of Indicator-SPP is intended to differentiate between distinct authentication methods, rather than to account for repeated use of the same method as used in Parikh-SPP. The following example highlights the difference between the two notions.
	
	\begin{figure}
		\centering
		\includegraphics[scale=1]{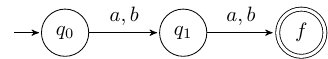}
		\caption{An instance of Parikh-SPP but not of Indicator-SPP.}
		\label{fig002}
	\end{figure}
	
	\begin{example}
		We consider the automaton of Figure~\ref{fig002}, where state $f$ is the only secret state, with $\ell(f) = 2$, and both events $a$ and $b$ are protectable, with 
		\(
			c(a) = c(b) = \gamma(a) = \gamma(b) = 1
		\).
		Let the protecting policy be $\mathcal{P} = \{a, b\}$. Then $\mathcal{P}$ is Parikh-valid and, in fact, Parikh-optimal, since there are four sequences reaching state $f$, namely $aa$, $ab$, $ba$, and $bb$, and for each sequence $w \in \{aa, ab, ba, bb\}$, we have 
		\(
			\gamma(a) \cdot |w|_a + \gamma(b) \cdot |w|_b = 2 \ge \ell(f).
		\)
		On the other hand, no indicator-valid protecting policy exists for this scenario. Indeed, as shown by the following calculation
		\(
			\gamma(a) \cdot \chi_{a}(aa) + \gamma(b) \cdot \chi_{b}(aa) = 1 < \ell(f),
		\)
		the sum of the clearance values is less than the required security threshold $\ell(f)$. Consequently, even the largest possible policy, $\mathcal P=\{a,b\}$, fails to meet the indicator-validity condition. Therefore, no protecting policy yields sufficient clearance to satisfy the security requirement under the indicator function.
		\hfill$\diamond$
	\end{example}
  
\section{Complexity of Parikh-SPP}
	We now show that Parikh-SPP(-U) is \NP-hard even under the restricted setting where both the cost and clearance functions are uniformly set to one, and the security requirement is binary, with exactly one state having a non-zero requirement. To this end, we prove \NP-hardness of the corresponding decision problem. Specifically, we show that BC-Parikh-SPP belongs to NP and that BC-Parikh-SPP-U is \NP-hard.

	\begin{lemma}\label{lemma1}
		BC-Parikh-SPP belongs to NP.
	\end{lemma}
	\begin{proof}
		Given an instance $\A$ over $\Sigma = \Sigma_p \uplus \Sigma_{up}$ of BC-Parikh-SPP, there are $2^{|\Sigma_p|}$ possible protecting policies. We can nondeterministically guess a correct policy $\mathcal P$ and verify, in polynomial time, that $\mathcal P$ is valid and satisfies the constraint $C(\mathcal P) \le W$. 
		To verify validity, we apply Dijkstra's algorithm to compute shortest paths from each initial state to all secret states, taking into account the clearances of the events specified by the policy~$\mathcal{P}$. If the clearance requirements of the secret states are met along the shortest path, then they are met along all paths~\cite{MaJC24}.
	\end{proof}

	Now, we prove \NP-hardness.
	\begin{lemma}\label{lemma2}
		BC-Parikh-SPP-U is \NP-hard. It remains \NP-hard even if the cost and clearance functions are constant with values one, and the security requirement satisfies $\ell(x)\in\{0,1\}$ where only a single state is secret, i.e., has a non-zero security requirement.
	\end{lemma}
	\begin{proof}
		To establish \NP-hardness, let $\varphi$ be a Boolean formula in 3-CNF with $n$ variables $x_1,\dots,x_n$ and $m$ clauses $C_1,\dots,C_m$. In each clause, we replace every literal $\neg x$ with $x'$. Without loss of generality, we may assume that the literals within each clause are distinct, that is, we may treat each clause as a set.

		We construct an automaton $\A=(Q,\Sigma,\delta,q_0,f)$ with a single initial state $q_0$ and a single accepting (secret) state $f$, as follows.
		\begin{figure}
			\centering
			\includegraphics[scale=1]{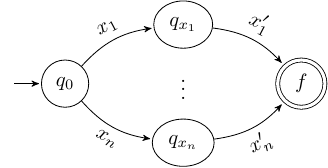}
			\caption{The first step of the reduction.}
			\label{fig1}
		\end{figure}
		For each variable $x$, we add to $\A$ two events, $x$ and $x'$, both of which are protectable, along with two transitions 
		\[
			(q_0,x,q_x) \text{ and } (q_x,x',f)\,, 
		\]
		where $q_x$ is a new state, see Figure~\ref{fig1}. This step introduces $n$ new states and $2n$ transitions.

		Now, for each clause $C_i=\{\lambda_{i,1},\dots,\lambda_{i,k_i}\}$, where $1\le i \le m$ and $1\le k_i \le 3$, we add $k_i$ transitions and $k_i-1$ new states to $\A$. Specifically, we construct the sequence of transitions
		\[
			(q_0,\lambda_{i,1},q_{i,1}), \dots, (q_{i,k_i\shortminus 1},\lambda_{i,k_i},f)\,,
		\]
		where $q_{i,1},\dots,q_{i,k_i\shortminus 1}$ are newly introduced intermediate states, see Figure~\ref{fig2} for an illustration, as well as Examples~\ref{exmp1} and~\ref{exmp2} below. This step introduces at most $2m$ new states and $3m$ transitions.

		\begin{figure}
			\centering
			\includegraphics[scale=1]{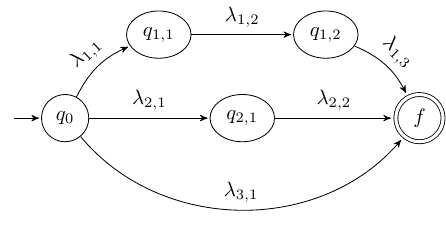}
			\caption{The encoding of three clauses $C_1=\{\lambda_{1,1},\lambda_{1,2},\lambda_{1,3}\}$, $C_2=\{\lambda_{2,1},\lambda_{2,2}\}$, and $C_3=\{\lambda_{3,1}\}$.}
			\label{fig2}
		\end{figure}

		We define the cost and clearance functions uniformly by setting $c(a) = 1$ and $\gamma(a) = 1$ for every event $a \in \Sigma$, meaning that all events in $\A$ are protectable. The security requirement function is defined by setting $\ell(f) = 1$ for a single designated secret state $f$, while all other states have security requirement zero, that is, they are not secret. 
		Finally, we set the budget limit $W$ to be the number of variables, that is, $W=n$.
		
		We now show that the formula $\varphi$ is satisfiable if and only if there is a Parikh-valid protection policy $\mathcal P$ with $C(\mathcal P) \le n$.

		If $\varphi$ is satisfiable, then there is a truth assignment to the variables that satisfies all clauses. Define $\mathcal P$ as the set of all literals that are assigned the value 1. For each variable $x$, the set $\mathcal P$ contains exactly one of $x$ and $x'$, ensuring that the paths in Figure~\ref{fig1} satisfy the security requirement of state $f$. Since the assignment satisfies every clause of $\varphi$, the set $\mathcal P$ contains at least one literal from each clause, thereby ensuring that the paths in Figure~\ref{fig2} also meet the security requirement of state $f$. Therefore, the set $\mathcal P$ forms a protecting policy of cardinality $n$, and thus satisfies the cost constraint $C(\mathcal P) \le n$.

		Conversely, suppose that $\mathcal P$ is a Parikh-valid protecting policy with $C(\mathcal P) \le n$. Since $\mathcal P$ must secure every path from $q_0$ to $f$, the structure in Figure~\ref{fig1} ensures that $\mathcal P$ contains at least one of the literals $x$ or $x'$ for each variable $x$. As there are $n$ such paths and the total cost is bounded by $n$, it follows that $\mathcal P$ contains exactly $n$ events. Therefore, for each variable $x$, the policy includes exactly one of $x$ and $x'$, but not both.
		This selection induces a consistent truth assignment for the variables of the formula $\varphi$: inclusion of $x$ in $\mathcal P$ corresponds to assigning $x=1$, while inclusion of $x'$ corresponds to $x=0$. Moreover, since $\mathcal P$ intersects all paths from $q_0$ to $f$ in Figure~\ref{fig2}, which correspond to the clauses of $\varphi$, the assignment satisfies every clause. Therefore, the policy $\mathcal P$ encodes a satisfying truth assignment for $\varphi$.
	\end{proof}
	
	The following examples illustrate the construction from the proof of Lemma~\ref{lemma2}. In the first example, we consider a satisfiable formula.
\begin{example}\label{exmp1}
	Let
	\(
		\varphi = (x \lor \neg y \lor z) \land (\neg x \lor y \lor \neg z)
	\)
	be a formula over variables $x, y, z$. The assignment $x=y=z=1$ satisfies both clauses, and therefore $\varphi$ is satisfiable.
	We construct an automaton $\A$ with events $\Sigma=\{x,x',y,y',z,z'\}$, all of which are protectable, i.e., $\Sigma_p=\Sigma$, and with the set of states
	\(Q=\{q_0,q_x,q_y,q_z,q_{1,1},q_{1,2},q_{2,1},q_{2,2},f\}\).
	For the variables $x,y,z$, we define the transitions 
	$(q_0, x, q_x)$, $(q_x, x', f)$, $(q_0, y, q_y)$, $(q_y, y', f)$, $(q_0, z, q_z)$, $(q_z, z', f)$, 
	and for the clauses $C_1=\{x, y', z\}$ and $C_2=\{x', y, z'\}$, we define the transitions 
	$(q_0, x, q_{1,1})$, $(q_{1,1}, y', q_{1,2})$, $(q_{1,2}, z, f)$, 
	and 
	$(q_0, x', q_{2,1})$, $(q_{2,1}, y, q_{2,2})$, $(q_{2,2}, z', f)$, 
	see Figure~\ref{fig_encoding_sat} for an illustration.
	The cost and clearance functions are defined by $c(a)=\gamma(a)=1$ for every $a\in \Sigma_p$, and the security requirement is given by $\ell(f)=1$ and $\ell(q)=0$ for all $q \ne f$. The budget $W=3$, the number of variables.
	Since $\varphi$ is satisfiable, there exists a Parikh-valid protecting policy $\mathcal P$. Indeed, the set \(\mathcal{P} = \{x, y, z\}\) intersects every path from $q_0$ to $f$ and satisfies the cost constraint $C(\mathcal{P}) = 3 \le W$. Therefore, $\mathcal{P}$ is a Parikh-valid protecting policy.
	\hfill$\diamond$
\end{example}

	\begin{figure}
		\centering
		\includegraphics[scale=1]{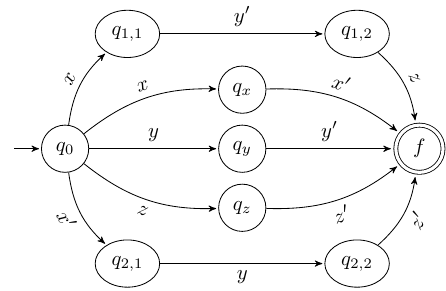}
		\caption{The encoding of $\varphi = (x \lor \neg y \lor z) \land (\neg x \lor y \lor \neg z)$.}
		\label{fig_encoding_sat}
	\end{figure}

	In the second example, we consider an unsatisfiable formula.
	\begin{example}\label{exmp2}
	Let
	\(
		\psi = (x \lor y \lor z) \land (\neg x) \land (\neg y) \land (\neg z)\,.
	\)
	The automaton obtained by the construction is shown in Figure~\ref{fig_encoding_unsat}. The alphabet consists of six events: $x,x',y,y',z,z'$. The edge labeled with $x',y',z'$ represents three parallel transitions from 	$q_0$ to $f$, corresponding to the three clauses $C_2=\{x'\}$, $C_3= \{y'\}$, and $C_4=\{z'\}$. The states include the initial state $q_0$, the secret state $f$, the three states $q_x$, $q_y$, and $q_z$ introduced for the variables, and the two intermediate states created for the first clause. Each event has cost and clearance~1, with $f$ being the only secret state with $\ell(f)=1$. The budget limit is $W=3$.
	In this automaton, any protecting policy within the budget must include exactly one of the events $x$ and $x'$, exactly one of the events $y$ and $y'$, and exactly one of the events $z$ and $z'$. This requirement leaves at least one path from $q_0$ to $f$ that avoids the protected events. Consequently, no policy can simultaneously satisfy both the security requirement and the budget constraint.
	\hfill$\diamond$
\end{example}

	\begin{figure}
		\centering
		\includegraphics[scale=1]{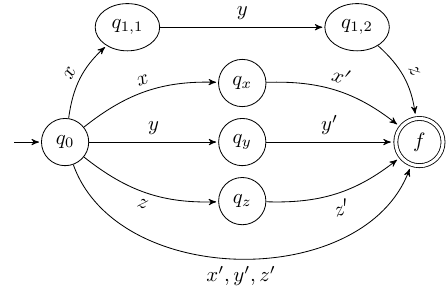}
		\caption{The encoding of $\psi = (x \lor y \lor z) \land (\neg x) \land (\neg y) \land (\neg z)$.}
		\label{fig_encoding_unsat}
	\end{figure}

	The preceding two lemmas yield the following main result.
	\begin{theorem}\label{thm1}
		BC-Parikh-SPP and BC-Parikh-SPP-U are \NP-complete, and Parikh-SPP and Parikh-SPP-U are \NP-hard. The problems remain \NP-hard even if the system has a unique initial state, the cost and clearance functions are constant (i.e., $c(x)=1$ and $\gamma(x)=1$), and the security requirement satisfies $\ell(x)\in \{0,1\}$, with exactly one state designated as secret.
	\end{theorem}
	\begin{proof}
		Lemma~\ref{lemma1} shows that BC-Parikh-SPP, and hence also BC-Parikh-SPP-U, belong to NP, while Lemma~\ref{lemma2} proves that the considered variants of the problem are \NP-hard. Strong \NP-hardness follows because the reduction of Lemma~\ref{lemma2} produces instances in which all numerical data are constant.
	\end{proof}

	Lemma~\ref{lemma1} suggests an algorithm whose runtime is exponential in the size of the alphabet of protectable events. We now show that, under current knowledge in complexity theory, this algorithm is optimal. Specifically, we show that if any of the considered budget-constrained secret protection problems could be solved in sub-exponential time, then 3-SAT could also be solved in sub-exponential time. 
	\begin{theorem}\label{thm3}
		There is no algorithm that solves BC-Parikh-SPP(-U) in time $2^{o(N)}$, where $N$ is the number of protectable events in the system, unless ETH fails.
	\end{theorem}
	\begin{proof}
		Consider the reduction of 3-SAT from the proof of Lemma~\ref{lemma2}. If the input 3-CNF formula contains $n$ variables, then the constructed system has $N = 2n$ protectable events. Suppose that there exists an algorithm solving BC-Parikh-SPP(-U) in time $2^{o(N)} = 2^{o(n)}$. Then, by reducing any instance of 3-SAT to an instance of BC-Parikh-SPP(-U) in polynomial time and applying this algorithm, we would obtain a sub-exponential-time algorithm for 3-SAT, contradicting the exponential time hypothesis.
	\end{proof}

		\begin{remark}
			It is worth noticing that BC-Indicator-SPP(-U) is also \NP-hard. In particular, the proof of Lemma~\ref{lemma2} shows that BC-Indicator-SPP-U is \NP-hard, which follows from the assumption that all literals within each clause are pairwise distinct. Consequently, there is no algorithm that would solve BC-Indicator-SPP(-U) in time $2^{o(N)}$, where $N$ is the number of protectable events, unless ETH fails.
			We analyze the exact complexity of BC-Indicator-SPP(-U) in Section~\ref{XSPP}.
		\end{remark}

	The intractability results characterize the inherent difficulty of the most challenging instances and do not necessarily reflect the performance on realistic inputs. They raise a fundamental question: \emph{Can SPP be solved efficiently in practice?}

	While Ma et al.~\cite{MaJC24} introduced a polynomial-time heuristic for Parikh-SPP, its scalability and robustness remain limited. In practice, instances with around 100 states may require several minutes to solve, and the approach offers no guarantees on solution quality. In the next section, we present an ILP-based algorithm for Parikh-SPP and show that it performs efficiently even on large-scale systems. In contrast, whether Indicator-SPP admits a practically efficient solution remains an open question.

\section{Solving Parikh-SPP via ILP}
	Integer Linear Programming (ILP) is used to determine the optimal outcome of a linear objective function, subject to linear constraints. Although the integrality constraints render ILP problems \NP-hard, advances in algorithms and solvers have made it practical to solve large and complex instances~\cite{Schrijver}. As a result, ILP has become a widely used tool for addressing a broad range of \NP-hard optimization problems.

	We design an algorithm that solves Parikh-SPP using an ILP solver. Given a candidate policy, its validity can be verified in polynomial time using Dijkstra's shortest paths algorithm~\cite{MaJC24}. In particular, if all shortest paths satisfy the security requirements of secret states, then all paths do as well. Thus, we can examine the policy that includes all protectable events and determine, in polynomial time, whether it satisfies the security constraints. If it does, the SPP instance is solvable; otherwise, it is not. Hence, without loss of generality, we assume that the input instance to Algorithm~\ref{algo1} is solvable.

\begin{algorithm}
\caption{Parikh-SPP solver using ILP}
\label{algo1}
\begin{algorithmic}[1]
	\Require A solvable instance of Parikh-SPP.
	\Ensure Optimal protecting policy $\mathcal P$.

    \State Define binary variables $x_e$ for each $e \in \Sigma_p$.
    \State Set objective function: $\text{Minimize } \sum_{e \in \Sigma_p} c(e) \cdot x_e$.

    \While{true}
        \State Solve the current ILP model.
        \State Get current values of $x_e$. 
        
        \For{each initial state $i$}
        	 \LComment{Verify the current solution.}
        	\State Compute the shortest paths from $i$ to all other states using Dijkstra's algorithm, considering current values of $x_e$ to adjust edge weights based on event clearance $\gamma$.
            \For{each secret state $s$}
            	  \State $d(i,s) \leftarrow$ shortest distance from $i$ to $s$
                \If{$d(i,s) < \ell(s)$}
                    \State add the constraint \(\sum_{e \in \pi} |\pi|_e \cdot \gamma(e) \cdot x_e \ge \ell(s)\) to the ILP model.
                \EndIf
            \EndFor
            \If{a new constraint was added to the ILP model}
            		\LComment{A violation of the solution was found, repeat the while-loop.}
                \State \textbf{break} (from outer for-loop)
             \Else 
             		\LComment{The found solution is a valid policy.}
             		\State \Return the current solution and \textbf{terminate}
            \EndIf
        \EndFor
    \EndWhile
\end{algorithmic}
\end{algorithm}

	The algorithm begins by introducing a binary variable for each protectable event: a value of 1 indicates that the event is included in the policy, and 0 indicates that it is not. The objective function, to be minimized, is defined as the total cost of the events selected for the policy, which is initially empty.
	An ILP solver is then used to compute a solution to the current ILP model. This solution is verified using Dijkstra's algorithm to ensure that all shortest paths satisfy the specified security requirements. If this condition holds, the solution is valid and the algorithm terminates.

	Otherwise, there is a shortest path $\pi$ that violates the security requirement of a secret state $s$. The algorithm then augments the ILP model with an additional constraint enforcing the security requirement along this path. Specifically, it ensures that the sum of the clearances of the events on $\pi$ meets or exceeds the required security level $\ell(s)$. This is expressed as $\sum_{e\in \pi} |\pi|_e \cdot \gamma(e) \cdot x_e \ge \ell(s)$, where $|\pi|_e$ denotes the number of occurrences of event $e$ in path $\pi$. After adding all such constraints to the ILP model, the algorithm proceeds to the next iteration of the while-loop. Since the instance is assumed to be solvable, the loop is guaranteed to eventually terminate with a valid solution.

	We implemented Algorithm~\ref{algo1} in Python, using \emph{PuLP} with the \emph{CBC} solver. All experiments were run on an Ubuntu 24.04.3 LTS machine with 32 Intel(R) Xeon(R) CPU E5-2660~v2 @ 2.20\,GHz and 246\,GB RAM. We ran up to 24 threads in parallel. We considered two data sources.
	(1) We derived Parikh-SPP instances from public automata benchmarks. In these instances, all events are protectable with $\gamma(\cdot)=c(\cdot)=1$ (this setting already makes the problem \NP-hard), and every accepting state is treated as a secret state with security level 1. Trivially unsolvable cases were filtered out.
	(2) We generated random NFAs using the Tabakov-Vardi model~\cite{TabakovV05}, using the MATA library~\cite{mata}, with $q$ states and the alphabet size $k$. For each combination of $q\in\{100, 500, 1000, 5000, 10000, 50000,\allowbreak 100000\}$ and $k\in\{2,3,5,10\}$, we generated 100 random instances, varying the transition density uniformly between $0.8$ and $5.0$. Each generated NFA was converted to a Parikh-SPP instance in which all events were protectable, and $\gamma(a)$, $c(a)$ and $\ell(f)$ were drawn at random from $\{1,...,10\}$. If an instance turned out to be trivially unsolvable, it was regenerated. In total, we solved 2800 random instances.

	For each instance, we logged the number of states, the optimal policy cost, the solver's status, the number of iterations, and the runtime. All instances, data, and scripts are available in the git repository \url{https://apollo.inf.upol.cz:81/vecerajakub/SPP-automata}.

	\begin{figure}
		\centering
		\includegraphics[scale=1]{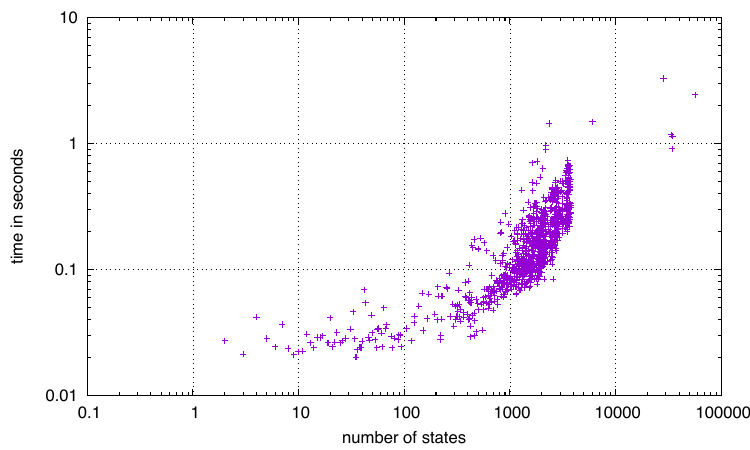}
		\caption{Runtime on real-derived Parikh-SPP instances.}
		\label{plot_real}
	\end{figure}

	For the real-based instances, Figure~\ref{plot_real} plots the runtime versus the number of states. All instances were solved quickly. Since the instances are not native Parikh-SPP data, they rather test structural difficulty than domain-specific Parikh-SPP semantics.

	For the random instances, our results show that the approach scales to NFAs with hundreds of thousands of states. Table~\ref{tab:generated} summarizes the median runtime performance over 100 runs for each size and alphabet combination (in seconds). Even for the largest instances with 100,000 states and 10 events, the median time was about 19 minutes, and the slowest 5\% of the instances were solved within 44 minutes. At moderate sizes (e.g., 10,000 states with $k=10$), the median runtime remained around 30 seconds. 
	
	For a fixed number of states, instances with a larger alphabet are generally harder, since more events and transitions increase the ILP size and the number of constraints generated.
	For example, for $q=\text{10,000}$ the median runtime rises from $2.56$s ($k=3$) to $30.28$s ($k=10$), and for $q=\text{100,000}$ from $179.32$s ($k=3$) to $1133.86$s ($k=10$). The few instances in the slowest 5\% typically had the highest transition density, which leads to more complex shortest-path computations and more constraint-generation iterations.

\begin{table}
\ra{1.25}
\centering
\caption{Median runtime in seconds across 100 runs.}
\label{tab:generated}
\begin{tabular}{rrrrr}
\toprule
$q$ states & $k{=}2$ & $k{=}3$ & $k{=}5$ & $k{=}10$ \\
\midrule100 & 0.039 & 0.081 & 0.247 & 1.783  \\
500 & 0.060  & 0.116  & 0.330  & 2.519  \\
1,000 & 0.087  & 0.152 & 0.442 & 2.780 \\
5,000 & 0.491 & 0.868  & 2.217 & 13.562  \\
10,000 & 1.593  & 2.563  & 5.323 & 30.275  \\
50,000 & 43.201  & 58.465 & 81.393 & 268.409 \\
100,000 & 163.626 & 179.318 & 228.184 & 1133.861  \\
\bottomrule
\end{tabular}
\end{table}

	One observation is the impact of multiple initial states on performance. For the randomly generated automata, the number of initial states grows with the model size (about 0.1\% of states were initial in those models). We ran a separate experiment generating similar random instances but restricting each automaton to a single initial state. The performance on these instances was noticeably better, especially for large systems. Table~\ref{tab:generated_oneinit} lists the median runtimes.
	We see that at 100,000 states and $k=10$, the median runtime drops to about 175s under the single-initial-state restriction, compared to 1134s in the general case. This trend is illustrated in Figure~\ref{plot_compare_median}, which compares the median running times for multi-initial vs. single-initial random instances. Intuitively, multiple initial states require the clearance constraints to be satisfied from each of those starting points, which yields a larger ILP and more constraints in the process of solving.

\begin{table}
\ra{1.25}
\centering
\caption{Random instances with one initial state---median runtime in seconds across 100 runs.}
\label{tab:generated_oneinit}
\begin{tabular}{rrrrr}
\toprule
$q$ states & $k{=}2$ & $k{=}3$ & $k{=}5$ & $k{=}10$ \\
\midrule100 & 0.044 & 0.078 & 0.261 & 1.657 \\
500 & 0.064 & 0.128 & 0.326  & 2.517 \\
1,000 & 0.078 & 0.151 & 0.458 & 2.671 \\
5,000 & 0.295 & 0.531 & 1.708 & 9.412 \\
10,000 & 0.642 & 1.237 & 3.328 & 18.764 \\
50,000 & 4.005  & 8.563  & 21.406 & 86.610 \\
100,000 & 9.458  & 17.079 & 34.733 & 174.971 \\
\bottomrule
\end{tabular}
\end{table}

	\begin{figure}
		\centering
		\includegraphics[scale=1]{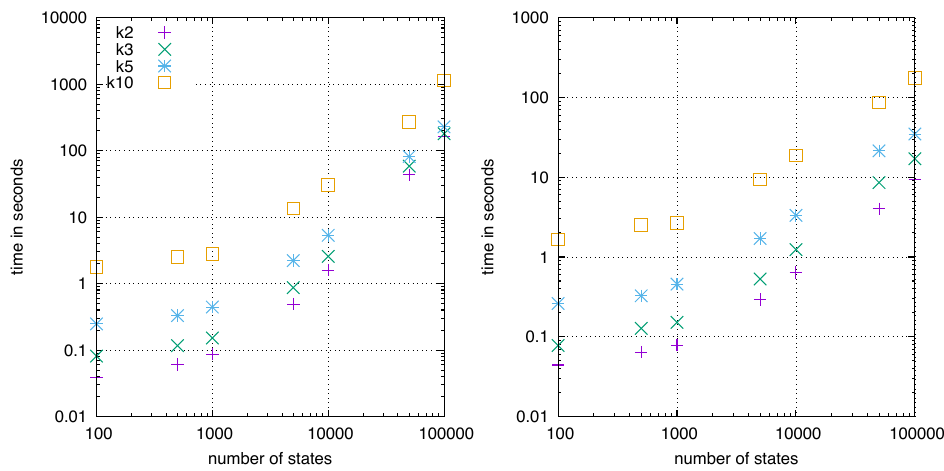}
		\caption{Median runtime on random instances with multiple initials (left) vs. a single initial state (right).}
		\label{plot_compare_median}
	\end{figure}

	In summary, our ILP-based approach scales to $10^5$ states and handles dense transitions. The runtime grows predictably with the number of states and alphabet size.
	Despite the underlying hardness of Parikh-SPP, these results indicate that Algorithm~\ref{algo1} provides an exact method for solving relatively large instances.

\section{Complexity of Indicator-SPP}\label{XSPP}
	In this section, we discuss BC-$\chi$-SPP where $\chi$ consists of indicator functions. We call this variant BC-Indicator-SPP, and prove that it is $\Sigma_{2}^{P}$-complete. To this end, we first show that verifying indicator-validity is computationally more challenging than verifying Parikh-validity.
 	\begin{proposition}\label{lemma001}
 		Given an automaton $\A$ and a policy $\mathcal P$, to verify whether $\mathcal P$ is indicator-valid is \coNP-complete.
 	\end{proposition}
	\begin{proof}
		To verify indicator-validity is to check there is no secret state $f$ and execution from an initial state to $f$ whose set of symbols $S \subseteq \mathcal P$  along the execution satisfies $\sum_{a\in S} \gamma(a) < \ell(f)$. To verify the complement is an NP search: we guess a state $f$ and a sequence $w$ of length at most the number of states and compute the weight of $w$ in polynomial time. Thus, to verify indicator-validity is in \coNP.
		
		To establish \coNP-hardness of the indicator-validity verification, let $\varphi$ be a Boolean formula in 3-CNF with $n$ variables $X=\{x_1,\dots,x_n\}$ and $m$ clauses $C_1,\dots,C_m$. We denote the negation of $x$ by $x'$. Without loss of generality, we assume that the literals in each clause are distinct, i.e., we may treat each clause as a set.
 		
 		From $\varphi$, construct the automaton $\A=(Q,\Sigma,\delta,q_0,\{q_{m+n}\})$ with a single initial state, $q_0$, and the set of states consisting of all clauses and all two-element sets of literals corresponding to a variable, that is, 
 		\(
 			Q = \{C_1,\dots,C_m\} \cup  \{ C_{x} =\{x,x'\}  \mid x\in X\}. 
 		\)
 		To simplify the presentation, we choose a fixed order of the states of $Q$, so that we can write $Q=\{q_0\} \cup \{q_1,\dots,q_{m+n}\}$, where $q_0$ is the initial state, $q_{i}$ is the clause $C_i$, for $i=1,\dots,m$, and $q_{m+1},\dots,q_{m+n}$ are the sets $C_{x_1},\dots,C_{x_n}$, respectively. We set the state $q_{m+n}$ secret with $\ell(q_{m+n})=n+1$; the other states are nonsecret.
 		The transitions are defined as follows: for every $x \in q_{i+1}$, we add the transition $(q_i,x,q_{i+1})$. Notice that all paths from $q_0$ to $q_{m+n}$ are of length $m+n$; see Figure~\ref{fig001b}.
		\begin{figure}
			\centering
			\includegraphics[scale=1]{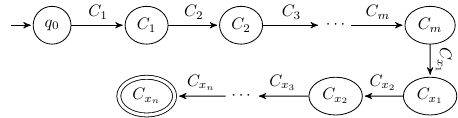}
			\caption{An illustration of the automaton $\A$.}
			\label{fig001b}
		\end{figure}
		
 		We now show that $\varphi$ is satisfiable if and only if the policy $\mathcal P = \bigcup_{x\in X} C_x$ is not indicator-valid.
 		
		If $\varphi$ is satisfiable, then there is an assignment to the variables that satisfies all clauses. Let $V$ be the set of literals that are assigned the value 1. Then $|V|=n$, and since $V$ satisfies all clauses, it contains at least one literal of every clause. Therefore, there is a string $w\in V^{m+n}$ leading from $q_0$ to $q_{m+n}$ that violates the requirement $\ell(q_{m+n}) = n+1$, since $w$ contains at most $|V|=n$ different events. Thus, $\mathcal P$ is not indicator-valid.
 		
		On the other hand, we assume that $\mathcal P$ is not indicator-valid, and consider a string $w$ over $\mathcal P$ from $q_0$ to $q_{m+n}$ that violates indicator-validity. Since $w$ leads through the states corresponding to the sets $C_{x}=\{x,x'\}$, for every $x\in X$, the string $w$ contains at least $n$ different events. These events correspond to the literals of $n$ different variables. Since $w$ violates indicator-validity, it contains exactly $n$ distinct events. However, these events must define a satisfying assignment for $\varphi$, because the string $w$ also goes through all the states that correspond to the clauses $C_1,\dots,C_m$, and hence the $i$th event of $w$ satisfies the clause $C_i$. Hence, $\varphi$ is satisfiable.
	\end{proof}
 	
	Recall that the class $\Sigma_2^P$ consists of problems that can be expressed as 
	\(
	L = \{w \mid \exists x \ \forall y \ R(w, x, y)\},
	\)
	where $R$ is a polynomial-time computable predicate. Intuitively, a language $L$ is in $\Sigma_2^P$ if there is a polynomial-time verifier that, given a guessed string $x$, can query an \NP oracle to verify that a condition over $y$ holds. An oracle Turing machine is a polynomial-time machine that can make queries to another decision problem (the \emph{oracle}) and receive instantaneous answers. A nondeterministic polynomial-time machine with access to an $\NP$ (or $\coNP$) oracle, denoted by $\NP^{\NP}$ (or $\NP^{\coNP}$), captures the computational power of $\Sigma_2^P$. We make use of this characterization in the proof of the main result of this section.
	
 	\begin{theorem}\label{thm_main}
 		 BC-Indicator-SPP(-U) is $\Sigma_2^P$-complete.
 	\end{theorem}
 	\begin{proof}
 		To show that BC-Indicator-SPP belongs to $\Sigma_2^P$, we guess a policy $\mathcal P \subseteq \Sigma_p$ with the cost not exceeding the budget constraint, and verify that $\mathcal P$ is indicator-valid. Since the guess of $\mathcal P$ can be done in nondeterministic polynomial time, the cost can be computed in polynomial time, and the check whether $\mathcal P$ is indicator-valid is a \coNP-complete problem by Proposition~\ref{lemma001}, we obtain that BC-Indicator-SPP belongs to $\text{NP}^{\text{coNP}} = \Sigma_2^P$.
 		
 		To prove hardness, we reduce the $\Sigma_2^P$-complete problem of 3-QSAT$_2$ \cite{Stockmeyer76}. The problem asks whether a formula of the form $\exists X\, \forall Y\, \varphi(X,Y)$, where $\varphi(X,Y)$ is a Boolean formula in 3-DNF and $X,Y$ are disjoint sets of variables, is satisfiable.
 		
 		To this end, let $\exists X\, \forall Y\, \varphi(X,Y)$ be a 3-QSAT$_2$ formula with $X=\{x_1,\dots,x_n\}$, $Y=\{y_1,\dots,y_r\}$, and with $m$ conjuncts $C_1,\dots,C_m$. In each conjunct, we replace every literal $\neg z$ with a fresh symbol $z'$. Without loss of generality, we assume that the literals within each conjunct are distinct, allowing us to treat each conjunct as a set. Additionally, we denote by $\overline{X}$ the set of all literals over $X$, that is, $\overline{X}=\{x,x' \mid x \in X\}$. Analogously, we define the set $\overline{Y}$ of $Y$-literals.
 		
		We construct an automaton $\A=(Q,\Sigma,\delta,q_0,\{f_1,f_2\})$ with a single initial state $q_0$ and two secret states $f_1$ and $f_2$. The alphabet of $\A$ consists of all literals, that is,  $\Sigma = \overline{X} \cup \overline{Y}$.
		The transition function of $\A$ is defined as follows. For each variable $x\in X$, we add two transitions 
		$
			(q_0,x,q_x) \text{ and } (q_x,x',f_1), 
		$
		where $q_x$ is a new state, and for each variable $y\in Y$, we add two transitions 
		\[
			(q_0,y,f_1) \text{ and } (q_0,y',f_1)\,,
		\]
		see Figure~\ref{fig1b} for an illustration. This step introduces $n$ new states and $2(n+r)$ transitions.
		
		\begin{figure}[t]
			\centering
			\includegraphics[scale=1]{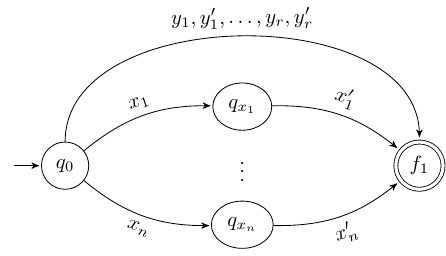}
     	\caption{The first step of the reduction.}
     	\label{fig1b}
		\end{figure}

		For each conjunct $C_i = \{\lambda_{i,1}, \dots, \lambda_{i,k_i}\}$, $1\le i \le m$, we add a state $C_i$ and $k_i$ transitions from $C_{i-1}$ to $C_i$ under the literals of $C_i$, that is, we add the transitions
		\[
				(C_{i-1},\lambda_{i,1},C_i), 	(C_{i-1},\lambda_{i,2},C_i), \dots, (C_{i-1},\lambda_{i,k_i},C_i)\,,
		\]
		where $C_0$ denotes the initial state $q_0$.
 		Finally, for every $y_i \in Y$, we add a new state $C_{y_i}$ together with two transitions from $C_{y_{i-1}}$ to $C_{y_i}$ under $y_i$ and $y_i'$, that is, we add 
 		\[
 			(C_{y_{i-1}},y_i,C_{y_i}) \text{ and } (C_{y_{i-1}},y_i',C_{y_i})\,,
 		\]
 		where $C_{y_0} = C_m$ and $C_{y_{r}}$ is denoted by $f_2$, see Figure~\ref{fig2b} for an illustration.
 		This step introduces $m+r$ new states and $2r + \sum_{i=1}^{m} k_i$ transitions, which is polynomial in the size of the formula.
     
     \begin{figure}[t]
     \centering
		 \includegraphics[scale=1]{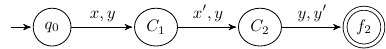}
     \caption{The encoding of the formula $\varphi = C_1 \lor C_2$ where $C_1=\{x, y\}$, $C_2=\{x', y\}$, $Y=\{y\}$, and $f_2=C_y$.}
     \label{fig2b}
     \end{figure}
     
 		To complete the construction, we define the cost and clearance functions uniformly by setting $c(a) = 1$ and $\gamma(a) = 1$ for every event $a \in \Sigma$, meaning that all events in $\A$ are protectable. The security requirement function is defined by setting $\ell(f_1) = 1$ and $\ell(f_2)=r+1$, while all the other states have security requirement zero. Finally, we set $W=2r+n$.
 		
 		We now show that the formula $\exists X\, \forall Y\, \varphi(X,Y)$ is satisfiable if and only if there is an indicator-valid protection policy $\mathcal P$ such that $C(\mathcal P) \le W$. To simplify the arguments, we state the following two claims.
 		
	\begin{claim}\label{lem:shape}
		Every indicator-valid policy $\mathcal P$ with $C(\mathcal P)\le W$ satisfies $C(\mathcal P)=2r+n$ and has the form
		\(
	  	\mathcal P = \overline{Y}\ \cup\ \{\text{exactly one of }x,x' \mid x\in X \}\,.
		\)
	\end{claim}
	\begin{proof}
		For a $Y$-literal $z$, a path from $q_0$ to $f_1$ is the single edge labelled $z$; since $\ell(f_1)=1$, indicator-validity forces $z\in\mathcal P$. Hence $\overline{Y}\subseteq\mathcal P$, contributing cost $2r$. Analogously, for each $x\in X$, a path from $q_0$ to $f_1$ is $q_0\to^{x} q_{x}\to^{x'} f_1$, which uses the two distinct events $x,x'$; as $\ell(f_1)=1$, $\mathcal P$ contains at least one of them, contributing cost at least $n$ in total. Thus, $C(\mathcal P)\ge 2r+n=W$, and the assumption $C(\mathcal P)\le W$ implies that $\mathcal P$ contains exactly one literal of each pair $\{x,x'\}$.
	\end{proof}
 		
 	Let $\nu_X$ be an assignment to the $X$-variables, and let
	\begin{equation}\label{eq:trueX}
	   \mathcal P = \overline{Y} \cup \{\lambda\in\overline{X} \mid \lambda \text{ is true under } \nu_X\}\,.
	\end{equation}

	\begin{claim}\label{lem:clear}
		The minimum, over all paths from $q_0$ to $f_2$, of the indicator clearance	$\bigl|\{\text{events of the path}\}\cap\mathcal P\bigr|$ equals $r$ if there is a $Y$-assignment $\nu_Y$ with $\varphi(\nu_X,\nu_Y)=0$, and is at least $r+1$ otherwise.
	\end{claim}
	\begin{proof}
		A path from $q_0$ to $f_2$ chooses a literal $\lambda_i\in C_i$ for $i=1,\dots,m$ and a literal $b_j\in\{y_j,y_j'\}$ for $j=1,\dots,r$. Its clearance is $\bigl|\{\lambda_1,\dots,\lambda_m,b_1,\dots,b_r\}\cap\mathcal P\bigr|$. The $r$ events $b_1,\dots,b_r$ are pairwise distinct and all lie in $\mathcal P$; therefore, they contribute exactly $r$. Let $\nu_Y$	be the $Y$-assignment that makes each chosen $b_j$ \emph{false}. We claim that $\lambda_i\in\mathcal P$ and $\lambda_i\notin\{b_1,\dots,b_r\}$ if and only if $\lambda_i$ is \emph{true} under $\nu_X \cup \nu_Y$: 
		If $\lambda_i$ is an $X$-literal, then $\lambda_i\notin\{b_1,\dots,b_r\}$, and by \eqref{eq:trueX} $\lambda_i\in\mathcal P$ iff $\lambda_i$ is true under $\nu_X$.
		If $\lambda_i$ is a $Y$-literal, then $\lambda_i\in\mathcal P$; and $\lambda_i\notin\{b_1,\dots,b_r\}$ iff $\lambda_i$ is true under $\nu_Y$.
		Hence the clearance equals $r$ plus the number of distinct chosen literals $\lambda_i$ that are true under $\nu_X \cup \nu_Y$.
		For a fixed $\nu_Y$, pick $\lambda_i\in C_i$ to be a literal false under $\nu_X \cup \nu_Y$ whenever one exists. Such a false literal exists iff $C_i$ is \emph{not} satisfied by $\nu_X \cup \nu_Y$. Therefore, the minimal clearance for $\nu_Y$ is exactly $r$ if every conjunct is unsatisfied, i.e.\ $\varphi(\nu_X,\nu_Y)=0$, and is at least $r+1$ otherwise.
	\end{proof}
		
	$(\Rightarrow)$ If the formula $\exists X\, \forall Y\, \varphi(X,Y)$ is satisfiable, then there is a truth assignment $\nu_X$ to the $X$-variables that satisfies the formula $\forall Y\, \varphi(X,Y)$. Let $\mathcal P=\overline{Y}\cup \{\lambda\in\overline{X} \mid \lambda \text{ is true under } \nu_X\}$, then $C(\mathcal P)=2r+n=W$ and $\mathcal P$ contains every $Y$-literal and one literal of each pair $\{x,x'\}$, that is, it satisfies the secret state $f_1$. Since no $\nu_Y$ makes $\varphi(\nu_X,\nu_Y)=0$, Claim~\ref{lem:clear} gives minimal clearance to $f_2$ at least $r+1=\ell(f_2)$, and hence $\mathcal P$ also satisfies $f_2$. Thus $\mathcal P$ is indicator-valid and within budget.
	
	$(\Leftarrow)$ Suppose $\mathcal P$ is indicator-valid with $C(\mathcal P)\le W$. By Claim~\ref{lem:shape}, $\mathcal P=\overline{Y}\cup\mathcal P_X$, where $\mathcal P_X$ contains exactly one literal of each pair $\{x,x'\}$. Let $\nu_X$ be the assignment defined by $\nu_X(x)=1$ if $x\in\mathcal P$ and $\nu_X(x)=0$ if $x'\in\mathcal P$; by~\eqref{eq:trueX}, $\mathcal P_X$ is exactly the set of $X$-literals true under $\nu_X$. As $\mathcal P$ satisfies $f_2$, every path to $f_2$ has clearance at least $r+1$, and hence by Claim~\ref{lem:clear} there is no $\nu_Y$ with $\varphi(\nu_X,\nu_Y)=0$; that is, $\nu_X$ satisfies $ \forall Y\, \varphi(X,Y)$, and the formula is satisfiable.
 	\end{proof}
 
	The following examples illustrate the construction.
 
 \begin{example}\label{ex:14}
	 Let
 	\(
    \varphi=\exists x_1,x_2 \forall y_1,y_2.(x_2\land\neg y_1)\lor(x_1\land y_2)\lor(\neg x_2\land\neg y_2)
 	\)
 	with conjuncts $C_1=\{x_2,y_1'\}$, $C_2=\{x_1,y_2\}$, $C_3=\{x_2',y_2'\}$. The assignment $\nu_X=(x_1=1,x_2=0)$ satisfies $\forall y_1,y_2\,\varphi$, and hence $\varphi$ is satisfiable. Here $n=r=2$, $m=3$, and the budget is $W=2r+n=6$. The automaton produced by the construction is shown in Figure~\ref{fig:ex14}; all events have cost and clearance $1$, $\ell(f_1)=1$, and $\ell(f_2)=r+1=3$.
	\begin{figure}
		\centering
		\includegraphics[scale=1]{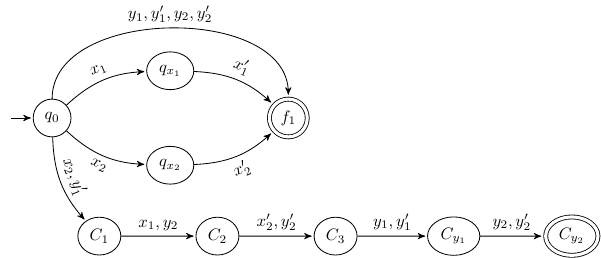}
		\caption{The automaton for $\varphi = \exists x_1,x_2$ $\forall y_1,y_2.$ $(x_2\land \neg y_1)$ $\lor$ $(x_1\land y_2)$ $\lor $ $(\neg x_2\land \neg y_2)$.}
		\label{fig:ex14}
	\end{figure}
	Consider the policy \(\mathcal P=\{x_1, x_2', y_1, y_1', y_2, y_2'\,\}\) corresponding to $\nu_X$ with \(C(\mathcal P)=6=W\). Under $\nu_X=(1,0)$, we have $C_1=x_2\land\neg y_1\equiv 0$, $C_2=x_1\land y_2\equiv y_2$, $C_3=\neg x_2\land\neg y_2\equiv\neg y_2$, and hence $C_1$ is never satisfied, while for every $\nu_Y$ exactly one of $C_2,C_3$ is satisfied, namely $C_2$ when $y_2=1$ and $C_3$ when $y_2=0$. The four universal assignments behave as discussed in Table~\ref{table001}.
	\begin{table}
	\caption{The four universal assignments.}
	\label{table001}
	\centering
	\begin{tabular}{cccc}
	\toprule
	$(y_1,y_2)$ & $(b_1,b_2)$ & Satisfied conjunct & Clearance \\
	\midrule
	$(0,0)$ & $(y_1,y_2)$   & $C_3$ ($x_2',y_2'$ true) & $2+1=3$ \\
	$(0,1)$ & $(y_1,y_2')$  & $C_2$ ($x_1,y_2$ true)   & $2+1=3$ \\
	$(1,0)$ & $(y_1',y_2)$  & $C_3$                    & $2+1=3$ \\
	$(1,1)$ & $(y_1',y_2')$ & $C_2$                    & $2+1=3$ \\
	\bottomrule
	\end{tabular}
	\end{table}
	Hence every path from $q_0$ to $f_2$ has clearance at least $3=\ell(f_2)$, that is, $\mathcal P$ is indicator-valid and meets the budget. 
	\hfill$\diamond$
 \end{example}

\begin{example}\label{ex:15}
	Let
	\(
	   \psi=\exists x_1,x_2 \forall y_1,y_2.(x_1\land y_1)\lor(x_2\land\neg y_1\land y_2)
	   \lor(\neg x_1\land\neg x_2\land y_2)
	\)
	with conjuncts $C_1=\{x_1,y_1\}$, $C_2=\{x_2,y_1',y_2\}$,	$C_3=\{x_1',x_2',y_2\}$. No assignment to $x_1,x_2$ satisfies $\psi$ for all	$y_1,y_2$, that is, $\psi$ is false. Here $n=r=2$, $m=3$, $W=6$, $\ell(f_1)=1$, and $\ell(f_2)=3$; the automaton is shown in	Figure~\ref{fig:ex15}.
	Every indicator-valid policy $\mathcal P$ with $C(\mathcal  P)\le W$ has the form	$\overline Y\cup \{\lambda\in\overline{X} \mid \lambda \text{ is true under } \nu_X\}$ for some assignment $\nu_X$ to the $X$-variables. We show that no indicator-valid policy within the budget exists by exhibiting a	path to $f_2$ of clearance $2<\ell(f_2)$.
	All four possible policies fail in the same way; a witnessing universal assignment and	the resulting clearance are listed below:
	\begin{table}[h]
	\centering
	\begin{tabular}{cccc}
	\toprule
	$\nu_X=(x_1,x_2)$ & $\mathcal P\setminus\overline Y$ & Witness $(y_1,y_2)$ & Clearance \\
	\midrule
	$(1,1)$ & $\{x_1,x_2\}$   & $(0,0)$ & $2$ \\
	$(1,0)$ & $\{x_1,x_2'\}$  & $(0,0)$ & $2$ \\
	$(0,1)$ & $\{x_1',x_2\}$  & $(1,0)$ & $2$ \\
	$(0,0)$ & $\{x_1',x_2'\}$ & $(0,0)$ & $2$ \\
	\bottomrule
	\end{tabular}
	\end{table}\\
	For instance, for $\nu_X=(1,1)$ and $(y_1,y_2)=(0,0)$, $C_1$ is falsified through $y_1$, $C_2$ through $y_2$, and $C_3$ through $x_1'\notin\mathcal P$, resulting in the path $y_1 y_2 x_1' y_1 y_2$ giving clearance $|\{y_1,y_2\}|=2$. Since none of these policies is valid, the construction	correctly reports $\psi$ as unsatisfiable.
	\hfill$\diamond$
\end{example}

\begin{figure}[t]
    \centering
		\includegraphics[scale=1]{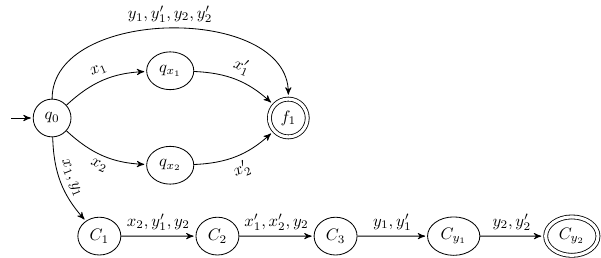}
    \caption{The automaton for $\psi = \exists x_1, x_2, \forall y_1, y_2. (x_1 \land y_1) \allowbreak \lor (x_2 \land \neg y_1 \land y_2 ) \lor (\neg x_1 \land \neg x_2 \land y_2)$.}
    \label{fig:ex15}
\end{figure}

\section{Conclusion}
	We addressed the computational complexity of two variants of the secret protection problem in discrete-event systems---Parikh-SPP and Indicator-SPP. While previous results showed that Parikh-SPP is solvable in polynomial time under the assumption of uniquely labeled transitions, we extended the understanding of its hardness by considering more general and practically relevant scenarios. 

	We showed that relaxing the uniqueness constraint on event labels makes Parikh-SPP \NP-hard. Moreover, we proved that the problem remains \NP-hard even if the cost and clearance functions are constant, and the security requirement is binary with only a single secret state. These results closed the gap left open in the literature and established \NP-hardness for the discussed variants of the problem. 

	We further showed that no sub-exponential-time algorithm can solve Parikh-SPP or Indicator-SPP unless the exponential time hypothesis fails. This result implies that, under current complexity-theoretic assumptions, exhaustive enumeration of all subsets of protectable events is essentially optimal.

	Given that the decision version of Parikh-SPP is \NP-com\-plete, we developed an ILP-based algorithm and conducted an extensive empirical evaluation demonstrating its scalability and efficiency on benchmark instances comprising up to hundreds of thousands of states.

	Finally, we analyzed the computational complexity of the decision version of Indicator-SPP, where only distinct protected events contribute to clearance. We proved that this problem is $\Sigma_{2}^{P}$-complete, placing it at the second level of the polynomial-time hierarchy. In contrast to Parikh-SPP, for which an ILP-based approach proved practically efficient, developing a practically efficient algorithm for Indicator-SPP remains an open challenge. Since the worst-case complexity need not reflect performance on realistic inputs, it remains open whether suitable techniques or heuristics could solve relevant instances of Indicator-SPP efficiently.

\begin{credits}
\subsubsection{\ackname}
	This research was supported by the Palacky University under the grants IGA PrF 2025 018 and IGA PrF 2026 026.
\end{credits}
\bibliographystyle{splncs04}
\bibliography{mybib}

\begin{thebibliography}{10}
\providecommand{\url}[1]{\texttt{#1}}
\providecommand{\urlprefix}{URL }
\providecommand{\doi}[1]{https://doi.org/#1}

\bibitem{AroraBarak2009}
Arora, S., Barak, B.: Computational Complexity - {A} Modern Approach. Cambridge
  University Press (2009)

\bibitem{mata}
Chocholat{\'{y}}, D., Fiedor, T., Havlena, V., Hol{\'{\i}}k, L., Hruska, M.,
  Leng{\'{a}}l, O., S{\'{\i}}c, J.: Mata: {A} fast and simple finite automata
  library. In: TACAS. LNCS, vol. 14571, pp. 130--151 (2024)

\bibitem{GareyJ79}
Garey, M.R., Johnson, D.S.: Computers and Intractability: {A} Guide to the
  Theory of NP-Completeness. W. H. Freeman (1979)

\bibitem{Hertli14}
Hertli, T.: 3-{SAT} faster and simpler---unique-{SAT} bounds for {PPSZ} hold in
  general. {SIAM} Journal on Computing  \textbf{43}(2),  718--729 (2014)

\bibitem{ImpagliazzoP01}
Impagliazzo, R., Paturi, R.: On the complexity of k-{SAT}. Journal of Computer
  and System Sciences  \textbf{62}(2),  367--375 (2001)

\bibitem{MaC22}
Ma, Z., Cai, K.: Optimal secret protections in discrete-event systems. {IEEE}
  Trans. Automatic Control  \textbf{67}(6),  2816--2828 (2022)

\bibitem{MaC25}
Ma, Z., Cai, K.: Secret protection in discrete-event systems with generalized
  confidentiality requirements. IEEE Trans. Automatic Control  \textbf{70}(4),
  2321--2333 (2025)

\bibitem{MaJC24}
Ma, Z., Jiang, J., Cai, K.: Secret protections with costs and disruptiveness in
  discrete-event systems using centralities. {IEEE} Trans. Automatic Control
  \textbf{69}(7),  4380--4395 (2024)

\bibitem{MatsuiC19}
Matsui, S., Cai, K.: Secret securing with multiple protections and minimum
  costs. In: 58th {IEEE} Conference on Decision and Control, {CDC} 2019, Nice,
  France, December 11-13, 2019. pp. 7635--7640 (2019)

\bibitem{PaturiPSZ05}
Paturi, R., Pudl{\'{a}}k, P., Saks, M.E., Zane, F.: An improved
  exponential-time algorithm for \emph{k}-{SAT}. Journal of the ACM
  \textbf{52}(3),  337--364 (2005)

\bibitem{Schrijver}
Schrijver, A.: Theory of linear and integer programming. Wiley-Interscience
  series in discrete mathematics and optimization (1999)

\bibitem{sipser1996introduction}
Sipser, M.: Introduction to the Theory of Computation. Cengage Learning (2012)

\bibitem{Stockmeyer76}
Stockmeyer, L.J.: The polynomial-time hierarchy. Theoretical Computer Science
  \textbf{3}(1),  1--22 (1976)

\bibitem{TabakovV05}
Tabakov, D., Vardi, M.Y.: Experimental evaluation of classical automata
  constructions. In: Logic for Programming, Artificial Intelligence, and
  Reasoning. LNCS, vol.~3835, pp. 396--411 (2005)

\end{thebibliography}
\end{document}